\IfFileExists{IEEEtran.cls}{
  \documentclass[conference]{IEEEtran}
}{
  \documentclass{paperreview}
}
\IEEEoverridecommandlockouts

\usepackage{cite}
\usepackage{amsmath,amssymb,amsfonts,amsthm,mathtools,bm}
\usepackage{graphicx}
\usepackage{textcomp}
\usepackage{xcolor}
\usepackage{booktabs}
\usepackage{array}
\usepackage{enumitem}
\usepackage{microtype}
\usepackage{placeins}
\usepackage{url}
\usepackage{hyperref}
\hypersetup{colorlinks=true,linkcolor=blue,citecolor=blue,urlcolor=blue}
\newtheorem{theorem}{Theorem}
\newtheorem{lemma}{Lemma}
\newtheorem{proposition}{Proposition}
\newtheorem{corollary}{Corollary}
\theoremstyle{remark}

\newcommand{\E}{\mathbb E}
\newcommand{\R}{\mathbb R}
\newcommand{\cW}{\mathcal W}
\newcommand{\gen}{\operatorname{gen}}
\newcommand{\hres}{\mathsf h}
\newcommand{\hb}{h_2}

\begin{document}

\title{Residual-Information Generalization Bounds\\
Under Selective Privacy: Extended Version}

\author{%
\begin{tabular}{ccc}
\begin{minipage}[t]{0.31\textwidth}
\centering
Amirreza Zamani\\
\textit{Department of Information Science and Engineering}\\
\textit{KTH Royal Institute of Technology}\\
Stockholm, Sweden\\
amizam@kth.se
\end{minipage}
&
\begin{minipage}[t]{0.31\textwidth}
\centering
Parastoo Sadeghi\\
\textit{School of Engineering and Technology}\\
\textit{UNSW Canberra}\\
Canberra, Australia\\
p.sadeghi@unsw.edu.au
\end{minipage}
&
\begin{minipage}[t]{0.31\textwidth}
\centering
Mikael Skoglund\\
\textit{Department of Information Science and Engineering}\\
\textit{KTH Royal Institute of Technology}\\
Stockholm, Sweden\\
skoglund@kth.se
\end{minipage}
\end{tabular}}

\maketitle
\ifdefined\paperreviewmode\thispagestyle{review}\fi

\begin{abstract}
Information-theoretic generalization bounds usually control dependence on the task sample, whereas selective privacy may constrain only a different, private coordinate. We study i.i.d. pairs $(X,Y)$, where a learner observes $(X^n,Y^n)$, leaks at most $\epsilon_n$ nats about $X^n$, and incurs loss on $Y$. A chain-rule transfer converts this mismatched constraint into the closed-form universal certificate
$I(W;Y^n)\leq\min\{nH(Y),nH(Y\mid X)+\epsilon_n\}$.
Combining the certificate with mutual-information generalization yields a bound governed by residual task uncertainty and normalized privacy leakage. We derive complementary task-compression and individual-sample refinements, identify a sharp obstruction showing why selective privacy alone need not give a vanishing certificate, and invert the bound to obtain the exact leakage allowance for a prescribed generalization tolerance. We prove exact information-level achievability for deterministic-task and independent-residual source families and show that the resulting $\sqrt{\epsilon_n/n}$ generalization rate is tight up to a universal constant. We also give a fixed-horizon adaptive application in which conditional leakage costs add.

\end{abstract}


\section{Introduction}

Information-theoretic generalization relates the dependence of a learned hypothesis on its training sample to the gap between empirical and population risk. For a sub-Gaussian loss, a standard estimate scales as $\sqrt{I(W;Z^n)/n}$~\cite{russo2016,xu2017}. Coordinatewise information can sharpen this estimate when a learner depends unevenly on different records~\cite{bu2020}. These principles are also central to adaptive data analysis and reusable holdout guarantees~\cite{dwork2015,bassily2016,steinke2020}.

The information controlled by a privacy requirement need not be the information needed by the risk bound. Consider one observation $Z=(X,Y)$. The coordinate $X$ is private, the coordinate $Y$ is task-relevant, and a randomized learner produces $W$ after observing $(X^n,Y^n)$. Selective information privacy imposes
\begin{equation}
 I(W;X^n)\leq\epsilon_n,                                      \label{eq:intro-privacy}
\end{equation}
while a loss evaluated on $Y$ is naturally controlled by $I(W;Y^n)$. Without a Markov restriction, neither data processing nor ordinary privacy generalization permits replacing the latter quantity by the former.

The missing term is the uncertainty in $Y$ that remains after $X$ is known. The chain rule gives
\begin{equation}
 I(W;Y^n)\leq I(W;X^n)+H(Y^n\mid X^n).                       \label{eq:intro-chain}
\end{equation}
For i.i.d. data, this produces the per-sample residual $H(Y\mid X)$. This term has a precise interpretation: privacy of $X^n$ cannot prevent a learner from retaining task randomness independent of $X^n$. The entropy cap $I(W;Y^n)\leq nH(Y)$ completes the certificate.

Our contributions are as follows. First, we combine this transfer argument with mutual-information generalization to obtain a closed-form selective-privacy bound. Second, we show that task compression removes the residual floor and that coordinatewise leakage yields an individual-sample refinement. Third, we establish construction-based tightness at two distinct levels: the transfer certificate is exactly attainable for deterministic-task and independent-residual source families, while a noisy-majority learner attains the generalization bound's $\sqrt{\epsilon_n/n}$ dependence up to a universal constant. Fourth, we derive the exact finite-sample leakage budget that certifies a target gap and record a deliberately restricted adaptive application for a fixed terminal iterate. %General checkpoint selection, stopping rules, and heterogeneous data arrivals remain outside the present scope.

The privacy notion in~\eqref{eq:intro-privacy} limits the mutual-information leakage. %and is not a differential-privacy guarantee.
Privacy-funnel mechanism design also studies retained utility under information leakage~\cite{makhdoumi2014,zamani2024}; here our objective is a learner-independent generalization certificate rather than an extremal privacy--utility function.

\section{Model and Information Transfer}

Let $(X,Y)\sim P_{XY}$ take values in finite alphabets, and let $(X_i,Y_i)_{i=1}^n$ be i.i.d. A randomized learner is any channel $P_{W\mid X^n,Y^n}$. For a loss $\ell:\cW\times\mathcal Y\to\R$, define
\begin{align}
 L_P(w)&:=\E_{Y\sim P_Y}[\ell(w,Y)],\nonumber\\
 L_n(w)&:=\frac1n\sum_{i=1}^n\ell(w,Y_i),\nonumber\\
 \gen(W,Y^n)&:=L_P(W)-L_n(W).                         \label{eq:risk}
\end{align}
For each fixed $w$, assume that $\ell(w,Y)-L_P(w)$ is $\sigma^2$-sub-Gaussian under $P_Y$. A loss in $[a,b]$ satisfies this assumption with $\sigma^2=(b-a)^2/4$.

The generalization module used below is the expected version of the mutual-information bound~\cite{xu2017}. It applies to any joint law of $(W,Y^n)$ generated from an i.i.d. task sample.

\begin{lemma}                                                        \label{lem:mi-gen}
Under the stated sub-Gaussian assumption,
\begin{equation}
 \left|\E[\gen(W,Y^n)]\right|
 \leq\sqrt{\frac{2\sigma^2}{n}I(W;Y^n)}.              \label{eq:mi-gen}
\end{equation}
\end{lemma}

For completeness, the mechanism behind Lemma~\ref{lem:mi-gen} is a change of measure. Under the decoupled law $P_WP_{Y^n}$, the centered empirical average
\begin{equation}
 F(W,Y^n):=L_n(W)-L_P(W)                              \label{eq:centered-average}
\end{equation}
is $\sigma^2/n$-sub-Gaussian. The variational representation of relative entropy therefore gives, for every $\lambda>0$,
\begin{align}
 \E_{P_{WY^n}}[F]
 &\leq \frac{D(P_{WY^n}\|P_WP_{Y^n})}{\lambda}
      +\frac{\lambda\sigma^2}{2n}\nonumber\\
 &=\frac{I(W;Y^n)}{\lambda}+\frac{\lambda\sigma^2}{2n}.             \label{eq:dv-step}
\end{align}
Optimizing $\lambda$ and repeating the argument for $-F$ yields~\eqref{eq:mi-gen}. This proof also shows why the dependence measure must involve the coordinate used by the loss: the decoupled reference law replaces $Y^n$, not an unrelated private block.

To connect selective privacy to Lemma~\ref{lem:mi-gen}, define the residual-information certificate
\begin{equation}
 \hres_\epsilon(P_{XY}):=\min\{H(Y),H(Y\mid X)+\epsilon\}.          \label{eq:h-def}
\end{equation}
This is an upper-bound function, not an operational supremum. It is nondecreasing, continuous, concave, and one-Lipschitz in $\epsilon$. It begins at $H(Y\mid X)$ and saturates at $H(Y)$ when $\epsilon\geq I(X;Y)$.

These structural properties can be summarized exactly. They are useful when comparing privacy budgets or averaging recordwise certificates.

\begin{proposition}                                                   \label{prop:h-properties}
For fixed $P_{XY}$, the map $\epsilon\mapsto\hres_\epsilon(P_{XY})$ is continuous, nondecreasing, concave, and one-Lipschitz. Moreover,
\begin{align}
 \hres_0(P_{XY})&=H(Y\mid X),\nonumber\\
 \hres_\epsilon(P_{XY})&=H(Y)\quad\text{if }\epsilon\geq I(X;Y).    \label{eq:h-properties}
\end{align}
\end{proposition}

\emph{Proof:}
The function equals the affine map $H(Y\mid X)+\epsilon$ until that map meets the constant $H(Y)$. The meeting point is $I(X;Y)$. Its slope is therefore one before saturation and zero afterward. \hfill$\square$

The next lemma is the bridge between the constrained coordinate and the task coordinate. It holds without a Markov condition on $(U,X,Y)$.

\begin{lemma}                                                        \label{lem:transfer}
For every jointly distributed $(U,X,Y)$ with finite-entropy $Y$,
\begin{equation}
 I(U;X)\leq\epsilon
 \quad\Longrightarrow\quad
 I(U;Y)\leq\hres_\epsilon(P_{XY}).                   \label{eq:transfer}
\end{equation}
The same statement holds for arbitrary blocks after replacing $(X,Y,\epsilon)$ by $(X^n,Y^n,\epsilon_n)$.
\end{lemma}

\emph{Proof:}
The chain rule yields
$I(U;Y)=I(U;X)+I(U;Y\mid X)-I(U;X\mid Y)$.
Discarding the last nonnegative term and bounding $I(U;Y\mid X)$ by $H(Y\mid X)$ gives $I(U;Y)\leq H(Y\mid X)+\epsilon$. Intersect this inequality with $I(U;Y)\leq H(Y)$. \hfill$\square$

The block implication does not require independent coordinates; only the later simplification $H(Y^n\mid X^n)=nH(Y\mid X)$ uses the i.i.d. law. It is also stable under post-processing. If $W'$ is obtained from $W$ and randomness independent of the sample given $W$, then data processing cannot increase either $I(W;X^n)$ or $I(W;Y^n)$. Thus the same certificate applies to quantization or randomized release of a trained model.

It is important that selective privacy does not authorize a full-sample substitution. The exact identity
\begin{equation}
 I(W;X^n,Y^n)=I(W;X^n)+I(W;Y^n\mid X^n)              \label{eq:full-sample-identity}
\end{equation}
contains a conditional task term that can be as large as $H(Y^n\mid X^n)$. Consequently, inserting $\epsilon_n$ for $I(W;X^n,Y^n)$ in a classical full-sample bound is generally invalid. Lemma~\ref{lem:transfer} controls the right object without making that substitution.

\section{Selective-Privacy Generalization}

The learner is called $(X^n,\epsilon_n)$-selectively private when~\eqref{eq:intro-privacy} holds. The constraint controls neither the full sample nor $Y^n$ directly. Applying Lemma~\ref{lem:transfer} at blocklength $n$ first quantifies the task information that any such learner can retain; Lemma~\ref{lem:mi-gen} then turns it into a risk certificate.

\begin{theorem}                                                       \label{thm:main}
Every $(X^n,\epsilon_n)$-selectively private learner satisfies
\begin{align}
 I(W;Y^n)&\leq\min\{nH(Y),nH(Y\mid X)+\epsilon_n\},                 \label{eq:main-info}\\
 \left|\E[\gen(W,Y^n)]\right|
 &\leq\sqrt{2\sigma^2\min\left\{H(Y),H(Y\mid X)+\frac{\epsilon_n}{n}\right\}}. \label{eq:main-gen}
\end{align}
\end{theorem}

\emph{Proof:}
Lemma~\ref{lem:transfer} and memorylessness give~\eqref{eq:main-info}. Substitution in~\eqref{eq:mi-gen} proves~\eqref{eq:main-gen}. \hfill$\square$

The entropy cap is active after the normalized privacy budget reaches $I(X;Y)$. Below that threshold, the certificate separates into the residual $H(Y\mid X)$ and the privacy cost $\epsilon_n/n$. If $Y=f(X)$, then $H(Y\mid X)=0$ and any $\epsilon_n=o(n)$ produces a vanishing bound. If the residual is positive and fixed, selective privacy alone does not certify convergence. This is a limitation of the available information constraint, not a lower bound on the generalization error of every learner.

The residual cannot be deleted from a universal task-information certificate. The following construction exposes the obstruction without invoking an extremal representation theorem.

\begin{proposition}                                                   \label{prop:obstruction}
Let $Y=(X,N)$, where $N$ is independent of $X$ and has finite entropy. The release $W=N^n$ satisfies $I(W;X^n)=0$ and
\begin{equation}
 I(W;Y^n)=nH(N)=nH(Y\mid X).                           \label{eq:obstruction}
\end{equation}
Consequently, no bound based only on $I(W;X^n)\leq\epsilon_n$ can replace $H(Y\mid X)$ by a uniformly smaller distribution-free term.
\end{proposition}

The proposition concerns retained task information. It does not assert that every selectively private algorithm has large risk or that mutual-information upper bounds are equalities for every loss.

\subsection{Achievability and Tightness}\label{subsec:tightness}

There are two separate questions. The first is whether the information-transfer certificate can be attained; the second is whether the induced generalization bound has the correct dependence on privacy and sample size. The next two results answer both questions for explicit, nontrivial families.

\begin{theorem}[Exact information-transfer achievability]\label{thm:transfer-tight}
The upper bound in Lemma~\ref{lem:transfer} is attained in each of the following regimes.
\begin{enumerate}[label=\textnormal{(\roman*)},leftmargin=*,nosep]
\item If $Y=f(X)$ and $0<H(Y)<\infty$, then for every $\epsilon\geq0$ there exists $U$ such that
\begin{align}
 I(U;X)&\leq\epsilon,\nonumber\\
 I(U;Y)&=\min\{H(Y),\epsilon\}.                       \label{eq:det-tight}
\end{align}
\item If $Y=(X,N)$, where $N$ is independent of $X$ and both entropies are finite, then for every $\epsilon\geq0$ there exists $U$ such that
\begin{align}
 I(U;X)&\leq\epsilon,\nonumber\\
 I(U;Y)&=\min\{H(Y),H(Y\mid X)+\epsilon\}.            \label{eq:residual-tight}
\end{align}
\end{enumerate}
The corresponding block bound~\eqref{eq:main-info} is also attained for every $n$ and every $\epsilon_n\geq0$ under the i.i.d. extensions of these sources.
\end{theorem}

\begin{proposition}[Rate tightness of the generalization bound]\label{prop:gen-tight}
Let $n$ be odd, let $X=Y$ be uniform on $\{-1,+1\}$, and put
$S_n=\sum_{i=1}^nY_i$ and $T_n=\operatorname{sign}(S_n)$. Given $T_n$, let the learner output $W\in\{-1,+1\}$ through a binary symmetric channel satisfying
\begin{equation}
 \Pr[W=T_n\mid T_n]=\frac{1+\eta}{2},\qquad 0\leq\eta\leq1.          \label{eq:noisy-majority}
\end{equation}
For the loss $\ell(w,y)=\bm1\{w\neq y\}$, which has sub-Gaussian proxy $\sigma^2=1/4$,
\begin{align}
 I(W;X^n)=I(W;Y^n)
 &=d(\eta),\label{eq:noisy-majority-info}\\
 \E[\gen(W,Y^n)]
 &=\frac{\eta}{2n}\E|S_n|,\label{eq:noisy-majority-gap}
\end{align}
where
\begin{equation}
 d(\eta)=\frac{1+\eta}{2}\log(1+\eta)
          +\frac{1-\eta}{2}\log(1-\eta).             \label{eq:d-eta}
\end{equation}
Consequently, along odd $n\to\infty$ and any $\eta_n\downarrow0$, with
$\epsilon_n=d(\eta_n)$,
\begin{equation}
 \E[\gen(W,Y^n)]
 =(1+o(1))\sqrt{\frac{\epsilon_n}{\pi n}}.            \label{eq:rate-tight}
\end{equation}
The right-hand side of~\eqref{eq:main-gen} is
$\sqrt{\epsilon_n/(2n)}$ in this example. Hence its dependence on
$\epsilon_n$ and $n$ is unimprovable in general, and the achieved-to-upper-bound ratio tends to $\sqrt{2/\pi}$.
Moreover, if
\begin{equation}
 q_n:=\frac{\eta}{n}\E|S_n|,
\end{equation}
then $I(W;X_i)=I(W;Y_i)=d(q_n)$ for every $i$. The right-hand side of the individual-sample bound~\eqref{eq:individual} is therefore
$\sqrt{d(q_n)/2}=(1+o(1))q_n/2$, whereas~\eqref{eq:noisy-majority-gap} equals $q_n/2$. Thus the individual-sample upper bound is asymptotically exact.
\end{proposition}

Theorem~\ref{thm:transfer-tight} is exact at the information level, whereas Proposition~\ref{prop:gen-tight} proves risk-level tightness up to a universal constant. Neither statement claims that~\eqref{eq:main-gen} is an equality for every source, learner, or loss. In particular, a positive residual term describes information permitted by selective privacy; a particular learner may retain less information and have a much smaller generalization gap. Full proofs appear in Appendix~K.

\section{Compression and Coordinatewise Refinements}

A direct task-information budget complements privacy by controlling precisely the component left unrestricted by~\eqref{eq:intro-privacy}. Suppose
\begin{equation}
 I(W;Y^n)\leq\rho_n.                                  \label{eq:task-rate}
\end{equation}
The stronger full-sample condition $I(W;X^n,Y^n)\leq\rho_n$ implies~\eqref{eq:task-rate} by data processing.

Intersecting the privacy route with the task-compression route gives the next result. The minimum makes clear that either route may dominate.

\begin{theorem}                                                       \label{thm:compression}
Under selective privacy and~\eqref{eq:task-rate},
\begin{equation}
 \left|\E[\gen]\right|
 \leq\sqrt{\frac{2\sigma^2}{n}
 \min\{\rho_n,nH(Y),nH(Y\mid X)+\epsilon_n\}}.       \label{eq:compression}
\end{equation}
\end{theorem}

Thus $\rho_n=o(n)$ removes the residual floor for every $P_{XY}$. A logarithmic task budget gives the familiar $O(\!\sqrt{\log n/n}\!)$ certificate. Privacy and compression therefore regularize different modes of memorization: the former limits sensitive dependence, while the latter limits residual task dependence.

Block information can be pessimistic when only a few records strongly influence the learner. The individual-sample information inequality~\cite{bu2020} states
\begin{equation}
 |\E[\gen]|
 \leq\frac1n\sum_{i=1}^n\sqrt{2\sigma^2 I(W;Y_i)}.    \label{eq:ismi}
\end{equation}
Transferring a separate privacy budget for each coordinate gives a data-dependent refinement of Theorem~\ref{thm:main}.

\begin{theorem}                                                       \label{thm:individual}
If $I(W;X_i)\leq\epsilon_i$ for every $i$, then
\begin{equation}
 |\E[\gen]|
 \leq\frac1n\sum_{i=1}^n
 \sqrt{2\sigma^2\min\{H(Y),H(Y\mid X)+\epsilon_i\}}. \label{eq:individual}
\end{equation}
\end{theorem}

\emph{Proof:}
Lemma~\ref{lem:transfer} applied to $(W,X_i,Y_i)$ bounds each $I(W;Y_i)$. Insert those bounds into~\eqref{eq:ismi}. \hfill$\square$

The block privacy constraint itself implies an averaged coordinatewise certificate. Independence of $X_1,\ldots,X_n$ gives
\begin{equation}
 \sum_{i=1}^n I(W;X_i)\leq I(W;X^n)\leq\epsilon_n.    \label{eq:block-individual}
\end{equation}
Since $u\mapsto\sqrt{\hres_u(P_{XY})}$ is concave and nondecreasing, Jensen's inequality reduces~\eqref{eq:individual} to~\eqref{eq:main-gen}. More importantly, nonuniform certified values $\epsilon_i$ can give a strictly smaller sum than replacing them all by their average.

\section{A Posterior-Variation Guarantee}

Mutual-information privacy is an average leakage constraint. Applications may instead certify that every released output changes the posterior of an individual private coordinate by only a small amount. Assume
\begin{equation}
 \|P_{X_i\mid W=w}-P_X\|_1\leq\delta_i
 \quad\text{for every $w$ in the support of $W$},      \label{eq:tv-condition}
\end{equation}
and let $p_{\min}:=\min_{x:P_X(x)>0}P_X(x)$. The finite alphabet and the restriction to the support of $P_X$ ensure that $p_{\min}>0$.

The posterior condition implies a coordinatewise mutual-information budget. Indeed, for any distributions $P$ and $Q$ on the same finite support,
\begin{align}
 D(P\|Q)&\leq \chi^2(P\|Q)
 =\sum_x\frac{(P(x)-Q(x))^2}{Q(x)}\nonumber\\
 &\leq\frac{\|P-Q\|_1^2}{\min_xQ(x)}.                \label{eq:kl-tv}
\end{align}
Applying~\eqref{eq:kl-tv} conditionally on $W=w$ and averaging over $W$ gives $I(W;X_i)\leq\delta_i^2/p_{\min}$. Combining this conversion with Theorem~\ref{thm:individual} yields a directly checkable posterior-variation result.

\begin{corollary}                                                     \label{cor:tv}
Under~\eqref{eq:tv-condition},
\begin{equation}
 |\E[\gen]|
 \leq\frac1n\sum_{i=1}^n\sqrt{2\sigma^2
 \min\left\{H(Y),H(Y\mid X)+\frac{\delta_i^2}{p_{\min}}\right\}}. \label{eq:tv-bound}
\end{equation}
\end{corollary}

This statement is intentionally posterior based. It should not be read as differential privacy: it depends on $P_X$, controls marginal records rather than neighboring datasets, and contains the smallest positive prior mass. Its value is that a uniform posterior certificate cannot hide a rare high-leakage output inside an average.

\section{Rates and Active Constraints}

The normalized budgets determine which information route is active. The next proposition records the limiting behavior of the displayed certificates and prevents a common misreading of a residual floor as a universal lower bound on risk.

\begin{proposition}                                                   \label{prop:rates}
Assume $H(Y)>0$ and let $a_n=\epsilon_n/n$. If $a_n\to a$, the right-hand side of~\eqref{eq:main-gen} converges to
\begin{equation}
 \sqrt{2\sigma^2\min\{H(Y),H(Y\mid X)+a\}}.          \label{eq:limit}
\end{equation}
The privacy-only certificate vanishes precisely when $H(Y\mid X)=0$ and $a_n\to0$. Under the task budget~\eqref{eq:task-rate}, the refined certificate also vanishes whenever $\rho_n/n\to0$, regardless of $H(Y\mid X)$.
\end{proposition}

The proof follows by continuity of the square root and minimum. The proposition classifies the certificate, not the realized gap. A learner that is independent of the task sample has zero expected gap even in a regime where~\eqref{eq:limit} is positive.

For quick comparison, Table~\ref{tab:rates} lists several common scalings. The deterministic-task rows assume $H(Y\mid X)=0$; the adaptive row additionally assumes a bounded per-update private-information cost.

\begin{table}[t]
\centering
\caption{Rates certified by the available information route.}
\label{tab:rates}
\begin{tabular}{@{}lll@{}}
\toprule
Condition & Budget scaling & Certified order \\
\midrule
Deterministic task & $\epsilon_n=O(1)$ & $O(n^{-1/2})$ \\
Deterministic task & $\epsilon_n=O(\log n)$ & $O(\sqrt{\log n/n})$ \\
Positive residual & $\epsilon_n=o(n)$ & residual floor \\
Task compression & $\rho_n=o(n)$ & $o(1)$ \\
$K$ fixed updates & $E_{n,K}=O(K)$ & $O(\sqrt{K/n})$ \\
\bottomrule
\end{tabular}
\end{table}

\section{Leakage Required for a Target Gap}

The preceding results map a certified privacy budget to a generalization tolerance. For design, the inverse question is often more useful: given a target $\gamma$, how much selective leakage can the certificate permit? Define the information tolerance
\begin{equation}
 \tau_\gamma:=\frac{\gamma^2}{2\sigma^2}.             \label{eq:tau-gamma}
\end{equation}
If $\tau_\gamma\geq H(Y)$, the entropy cap already certifies the target. The nontrivial regime is therefore $\tau_\gamma<H(Y)$, where the privacy and task-compression branches must carry the guarantee. Monotonicity of the square root makes their inverse thresholds exact for the displayed certificate.

\begin{corollary}                                                     \label{cor:inverse-design}
Fix $\gamma>0$ with $\tau_\gamma<H(Y)$. The privacy route in Theorem~\ref{thm:main} certifies
$|\E[\gen]|\leq\gamma$ if and only if
\begin{equation}
 H(Y\mid X)\leq\tau_\gamma,\qquad
 \epsilon_n\leq n\bigl(\tau_\gamma-H(Y\mid X)\bigr).             \label{eq:inverse-privacy}
\end{equation}
Under the direct task budget~\eqref{eq:task-rate}, the combined certificate meets the target if and only if either~\eqref{eq:inverse-privacy} holds or
\begin{equation}
 \rho_n\leq n\tau_\gamma.                              \label{eq:inverse-task}
\end{equation}
For a fixed leakage budget $\epsilon_n=\epsilon$ and $\tau_\gamma>H(Y\mid X)$, the privacy route therefore succeeds for every
\begin{equation}
 n\geq\left\lceil\frac{\epsilon}
 {\tau_\gamma-H(Y\mid X)}\right\rceil.                \label{eq:inverse-sample}
\end{equation}
No sample size can make that route certify the target when $\tau_\gamma\leq H(Y\mid X)$ and $\epsilon>0$.
\end{corollary}

The residual thus has a concrete design meaning. It consumes $H(Y\mid X)$ nats of the per-sample information tolerance before any algorithm-dependent leakage can be admitted. In the deterministic case,~\eqref{eq:inverse-sample} reduces to $n\geq\lceil2\sigma^2\epsilon/\gamma^2\rceil$; with residual task randomness, task compression can meet tolerances that privacy alone cannot certify.

\section{A Fixed-Horizon Adaptive Application}

Repeated updates can reveal more than the final update viewed in isolation. We therefore account for the complete visible history. At round $k$, let $V_k$ be generated from $(X^n,Y^n)$ and the preceding transcript $\Pi_{k-1}=(V_1,\ldots,V_{k-1})$. Fix the horizon $K$, and assume deterministic conditional budgets
\begin{align}
 I(V_k;X^n\mid\Pi_{k-1})&\leq e_{n,k},               \label{eq:ledger-x}\\
 I(V_k;Y^n\mid\Pi_{k-1})&\leq r_{n,k}.               \label{eq:ledger-y}
\end{align}
Let $W_K$ be a randomized post-processing of the full transcript using no additional data access. This formulation charges every released model, diagnostic, and data-dependent control signal that is visible to later rounds.

The chain rule converts the roundwise constraints into one block constraint. Applying the one-shot results to the transcript yields the following limited adaptive consequence.

\begin{corollary}                                                     \label{cor:adaptive}
Write $E_{n,K}=\sum_{k=1}^K e_{n,k}$ and $R_{n,K}=\sum_{k=1}^K r_{n,k}$. Under~\eqref{eq:ledger-x},
\begin{align}
 G_{n,K}&:=\min\left\{H(Y),H(Y\mid X)+\frac{E_{n,K}}n\right\},\nonumber\\
 |\E[\gen(W_K,Y^n)]|
 &\leq\sqrt{2\sigma^2G_{n,K}}.                       \label{eq:adaptive-one}
\end{align}
If~\eqref{eq:ledger-y} also holds, $R_{n,K}/n$ is an additional branch in the minimum.
\end{corollary}

Indeed, the chain rule gives $I(\Pi_K;X^n)\leq E_{n,K}$ and $I(\Pi_K;Y^n)\leq R_{n,K}$; data processing transfers both inequalities to $W_K$. For $Y=f(X)$ and constant per-update leakage, the privacy route vanishes when $K=o(n)$. A summable schedule can keep $E_{n,K}$ bounded for an arbitrarily long predetermined horizon. The result does not cover data-dependent checkpoint selection, stopping, or changing source laws; those require a broader transcript theorem.

\section{Composition Diagnostics and Budget Schedules}

Conditional accounting in~\eqref{eq:ledger-x} is essential. Small unconditional leakage for every separate release does not imply small leakage for their joint transcript. To see this, let a protected bit $B$ and a random key $R$ be independent and uniform. Release $V_1=R$ and $V_2=B\oplus R$. Each marginal release is independent of $B$, so
\begin{equation}
 I(V_1;B)=I(V_2;B)=0,                                \label{eq:synergy-marginal}
\end{equation}
but the pair reconstructs the bit and satisfies
\begin{equation}
 I(V_1,V_2;B)=H(B),\qquad I(V_2;B\mid V_1)=H(B).      \label{eq:synergy-joint}
\end{equation}
The conditional ledger detects exactly the information that becomes visible after earlier releases are known. This simple synergy example rules out summing marginal per-round guarantees.

The ledger also makes the effect of update schedules explicit. Suppose $Y=f(X)$ and the task ledger is unavailable. The privacy route in~\eqref{eq:adaptive-one} is
\begin{equation}
 |\E[\gen(W_K,Y^n)]|\leq\sqrt{2\sigma^2E_{n,K}/n}.   \label{eq:adaptive-deterministic}
\end{equation}
Several common schedules lead to qualitatively different update capacities:
\begin{align}
 e_{n,k}=e_0&\quad\Longrightarrow\quad E_{n,K}=e_0K,                 \label{eq:constant-schedule}\\
 e_{n,k}=e_0\beta^{k-1}&\quad\Longrightarrow\quad
 E_{n,K}\leq\frac{e_0}{1-\beta},\quad 0<\beta<1,                    \label{eq:geometric-schedule}\\
 e_{n,k}=\frac{e_0}{k}&\quad\Longrightarrow\quad
 E_{n,K}\leq e_0(1+\log K).                                        \label{eq:harmonic-schedule}
\end{align}
Thus constant-cost updates require $K=o(n)$ for the displayed certificate to vanish. Geometrically decaying costs give an $O(n^{-1/2})$ certificate uniformly over every predetermined $K$. Harmonic costs permit $K$ to grow subexponentially in $n$ while maintaining $E_{n,K}=o(n)$. More generally, a schedule $e_{n,k}=e_0k^{-p}$ is summable for $p>1$ and has logarithmic growth at $p=1$.

These calculations concern certified conditional mutual informations, not nominal optimizer step sizes. A small learning rate need not imply a small information cost. Conversely, many internal optimization steps are free if they are not released and the final channel can be certified directly. The transcript model is appropriate when intermediate values influence later data-dependent choices or are exposed outside the training routine.

Task compression composes in the same manner. If a stronger conditional full-sample budget
\begin{equation}
 I(V_k;X^n,Y^n\mid\Pi_{k-1})\leq c_{n,k}              \label{eq:full-ledger}
\end{equation}
is available, conditional data processing implies both~\eqref{eq:ledger-x} and~\eqref{eq:ledger-y} with $e_{n,k}=r_{n,k}=c_{n,k}$. The terminal certificate then uses
\begin{align}
 &\min\left\{H(Y),H(Y\mid X)+\frac{\sum_k c_{n,k}}n,
                         \frac{\sum_k c_{n,k}}n\right\}\nonumber\\
 &\hspace{2em}=\min\left\{H(Y),\frac{\sum_k c_{n,k}}n\right\}.       \label{eq:full-ledger-bound}
\end{align}
In that case the direct task route dominates the residual privacy route. Reporting both routes remains useful because selective privacy can be certified even when a full-sample information guarantee is unavailable.

An audit should therefore record the random object released at each round, condition its cost on the complete preceding history, and separate private-coordinate leakage from task information. Only after summing the corresponding ledgers should the residual transfer be applied. This order avoids treating selective privacy as if it were full-sample stability and makes the scaling with the number of updates transparent.

\section{Loss-Aware Task Reduction}

The task coordinate should be chosen at the granularity actually used by the loss. Suppose there is a finite-valued statistic $R=\phi(Y)$ and a loss $\widetilde\ell$ such that
\begin{equation}
 \ell(w,Y)=\widetilde\ell(w,R).                       \label{eq:task-reduction}
\end{equation}
The learner may still use the full pair $(X^n,Y^n)$; no Markov restriction is imposed on $W$. Nevertheless, empirical and population risk depend only on $R^n$. Applying the transfer argument to the pair $(X,R)$ yields the sharper loss-aware certificate.

\begin{corollary}                                                     \label{cor:task-reduction}
If $\widetilde\ell(w,R)$ is $\sigma^2$-sub-Gaussian and
$I(W;X^n)\leq\epsilon_n$, then
\begin{align}
 I(W;R^n)&\leq\min\{nH(R),nH(R\mid X)+\epsilon_n\},    \label{eq:reduced-info}\\
 |\E[\gen(W,R^n)]|
 &\leq\sqrt{2\sigma^2\min\left\{H(R),H(R\mid X)+\frac{\epsilon_n}{n}\right\}}.
                                                               \label{eq:reduced-gen}
\end{align}
\end{corollary}

The proof is Theorem~\ref{thm:main} with $Y$ replaced by $R$. Because $R$ is a deterministic function of $Y$,
\begin{equation}
 H(R\mid X)\leq H(Y\mid X),\qquad H(R)\leq H(Y).       \label{eq:coarsening-entropies}
\end{equation}
Thus a loss-aware reduction never enlarges either entropy branch. The improvement can be strict when the raw label includes nuisance information that does not affect evaluation. For several losses, one may take $R$ to be the tuple of all label statistics they use and certify the joint task.

Conditional entropy may be bounded through prediction error when it is easier to certify a decoder than a full source law. Let $|\mathcal Y|=m\geq2$, let $g:\mathcal X\to\mathcal Y$, and define
\begin{equation}
 p_e:=\Pr\{g(X)\neq Y\}.                              \label{eq:prediction-error}
\end{equation}
Introduce the error indicator $E=\mathbf 1\{g(X)\neq Y\}$. Since $E$ is determined by $(X,Y)$,
\begin{align}
 H(Y\mid X)
 &\leq H(E,Y\mid X)\nonumber\\
 &=H(E\mid X)+H(Y\mid E,X)\nonumber\\
 &\leq \hb(p_e)+p_e\log(m-1).                         \label{eq:fano-upper}
\end{align}
The last line uses that $Y=g(X)$ when $E=0$ and has at most $m-1$ possibilities when $E=1$. Substitution gives a predictability-based form of the main result.

\begin{corollary}                                                     \label{cor:fano}
Under the assumptions of Theorem~\ref{thm:main}, every decoder $g$ with error $p_e$ gives
\begin{align}
 |\E[\gen]|
 &\leq\left[2\sigma^2\min\left\{H(Y),\right.\right.\nonumber\\[-1mm]
 &\hspace{12mm}\left.\left.
 \hb(p_e)+p_e\log(m-1)+\frac{\epsilon_n}{n}\right\}\right]^{1/2}.
 \label{eq:fano-gen}
\end{align}
\end{corollary}

The decoder need not be used by the learning algorithm; it only certifies that the private coordinate predicts the task. At $p_e=0$, the deterministic-task rate is recovered. A fixed positive error produces a residual floor, consistent with Proposition~\ref{prop:obstruction}. The same reasoning applies after task reduction by predicting $R=\phi(Y)$ directly, often with a smaller alphabet and a lower error.

For example, a ten-class task with a certified error $p_e=0.02$ satisfies
$H(Y\mid X)\leq \hb(0.02)+0.02\log9\approx0.142$ nats. For $[0,1]$ loss and vanishing normalized leakage, the corresponding privacy-only certificate is at most $\sqrt{0.142/2}\approx0.266$. This numerical value remains an upper certificate, not a claim that the realized gap approaches the floor.

\section{Binary Illustration and Design Lessons}

Let $X\sim\operatorname{Bernoulli}(1/2)$ and
$Y=X\oplus N$, where $N\sim\operatorname{Bernoulli}(\theta)$ is independent and $0\leq\theta\leq1/2$. With natural logarithms,
\begin{equation}
 H(Y)=\log2,\qquad H(Y\mid X)=\hb(\theta),           \label{eq:binary-entropy}
\end{equation}
where $\hb(u)=-u\log u-(1-u)\log(1-u)$. For a $[0,1]$-valued loss, Theorem~\ref{thm:main} becomes
\begin{equation}
 |\E[\gen]|
 \leq\sqrt{\frac12\min\left\{\log2,\hb(\theta)+\frac{\epsilon_n}{n}\right\}}. \label{eq:binary-bound}
\end{equation}
At $\theta=0$, the task is determined by the private coordinate and selective privacy directly controls task memorization. For every fixed $\theta>0$, the privacy-only certificate approaches $\sqrt{\hb(\theta)/2}$ as $\epsilon_n/n\to0$. Adding a task budget $\rho_n$ replaces the interior minimum by
$\min\{\log2,\hb(\theta)+\epsilon_n/n,\rho_n/n\}$ and restores a vanishing rate whenever $\rho_n=o(n)$.

The privacy-only limiting values are already substantial for modest task noise. In nats, $\hb(0.01)\approx0.056$, $\hb(0.05)\approx0.199$, and $\hb(0.10)\approx0.325$. For bounded loss, these correspond to limiting certificates approximately $0.167$, $0.315$, and $0.403$. The rapid increase reflects a structural change rather than a loose dependence on $n$: the noise sequence $N^n$ is a task component that selective privacy does not constrain.

The example also distinguishes a certificate from an algorithmic guarantee. A nonzero residual says that the privacy constraint permits information-rich outputs; a particular learner may still generalize because it retains much less than the permitted information. Practical use therefore calls for auditing both leakage about $X^n$ and direct dependence on $Y^n$, especially after repeated tuning.

Three conclusions follow. First, a selective constraint must be transferred to the coordinate on which loss is evaluated before invoking an information-theoretic generalization theorem. Second, the conditional entropy is structural, not a proof artifact: it measures an unconstrained task channel, and Theorem~\ref{thm:transfer-tight} shows that the resulting information certificate is exactly attainable for broad source families. Third, privacy, task compression, and coordinatewise accounting are complementary controls. This focused treatment intentionally leaves the unified privacy--fairness framework and more general adaptive theory to a longer journal development.

\section{Conclusion}

Selective mutual-information privacy can imply a useful generalization certificate even when the protected and task coordinates differ. The correct bridge is a residual-information inequality that combines normalized leakage with $H(Y\mid X)$ and respects the entropy cap. The transfer certificate is exactly attainable for deterministic and independent-residual source families, and an explicit learner shows that the deterministic-task generalization rate is tight up to a universal constant. Residual task randomness requires an additional information bottleneck. Coordinatewise budgets and fixed-horizon information ledgers extend the same principle without changing its interpretation.

%\clearpage
\section*{Appendix: Proofs and Analytical Details}

This appendix contains only proofs and derivations of results stated in the main text. The statements are not repeated.

\subsection*{A. Proof of Lemma~\ref{lem:mi-gen}}

Let $P=P_{WY^n}$ and $Q=P_WP_{Y^n}$. Define
\begin{equation}
 F(w,y^n):=L_P(w)-\frac1n\sum_{i=1}^n\ell(w,y_i).      \label{eq:app-F}
\end{equation}
For fixed $w$, put $\xi_w(Y)=L_P(w)-\ell(w,Y)$. Independence of
$Y_1,\ldots,Y_n$ and the sub-Gaussian assumption give
\begin{align}
 \log\E_{Y^n}e^{\lambda F(w,Y^n)}
 &=\sum_{i=1}^n\log\E_Y
   \exp\left\{\frac{\lambda}{n}\xi_w(Y)\right\}\nonumber\\
 &\leq \frac{\lambda^2\sigma^2}{2n}.                 \label{eq:app-mgf}
\end{align}
The same bound holds after averaging over $W\sim P_W$, so it holds under $Q$. The Donsker--Varadhan variational inequality implies, for every $\lambda>0$,
\begin{align}
 \lambda\E_P[F]
 &\leq D(P\|Q)+\log\E_Q e^{\lambda F}\nonumber\\
 &\leq I(W;Y^n)+\frac{\lambda^2\sigma^2}{2n}.        \label{eq:app-dv}
\end{align}
When $I(W;Y^n)>0$, minimize after division by $\lambda$ at
$\lambda=\sqrt{2nI(W;Y^n)/\sigma^2}$. This yields
\begin{equation}
 \E_P[F]\leq\sqrt{\frac{2\sigma^2}{n}I(W;Y^n)}.      \label{eq:app-one-side}
\end{equation}
Repeating the calculation with $-F$ controls the opposite sign. When the mutual information is zero, $P=Q$ on the relevant sigma field and the expected centered quantity is zero.

The same argument also proves the individual-sample inequality used in~\eqref{eq:ismi}. For each $i$, apply the one-sample change of measure to $(W,Y_i)$ and $L_P(W)-\ell(W,Y_i)$. Summation and the triangle inequality give
\begin{equation}
 |\E[\gen(W,Y^n)]|
 \leq\frac1n\sum_{i=1}^n\sqrt{2\sigma^2I(W;Y_i)}.     \label{eq:app-ismi}
\end{equation}

\subsection*{B. Proof of Proposition~\ref{prop:h-properties} and Lemma~\ref{lem:transfer}}

The certificate has the equivalent representation
\begin{equation}
 \hres_\epsilon(P_{XY})
 =H(Y\mid X)+\min\{\epsilon,I(X;Y)\}.                 \label{eq:app-h-piecewise}
\end{equation}
The second term is affine with slope one on $[0,I(X;Y)]$ and constant thereafter. It is therefore continuous, nondecreasing, concave, and one-Lipschitz. Substitution at zero gives $H(Y\mid X)$, while every $\epsilon\geq I(X;Y)$ gives $H(Y)$. This proves Proposition~\ref{prop:h-properties}.

For the transfer implication, expand the joint mutual information in two orders:
\begin{align}
 I(U;X,Y)&=I(U;X)+I(U;Y\mid X)\nonumber\\
         &=I(U;Y)+I(U;X\mid Y).                      \label{eq:app-chain-two}
\end{align}
Rearranging and applying the assumptions gives
\begin{align}
 I(U;Y)
 &=I(U;X)+I(U;Y\mid X)-I(U;X\mid Y)\nonumber\\
 &\leq\epsilon+H(Y\mid X).                           \label{eq:app-transfer-residual}
\end{align}
The universal inequality $I(U;Y)\leq H(Y)$ supplies the other branch. No conditional independence assumption has been used. Replacing $(X,Y)$ by arbitrary blocks proves the block statement.

Equality in the residual branch requires saturation of the leakage budget together with
$I(U;Y\mid X)=H(Y\mid X)$ and $I(U;X\mid Y)=0$. These conditions need not be simultaneously feasible. This explains why $\hres_\epsilon$ is used as an upper certificate rather than as an achievability claim.

If $U'$ is generated from $U$ and fresh randomness conditionally independent of $(X,Y)$ given $U$, the Markov chains $U'\!-!U\!-!X$ and $U'\!-!U\!-!Y$ give
$I(U';X)\leq I(U;X)$ and $I(U';Y)\leq I(U;Y)$. This proves the post-processing observation.

\subsection*{C. Proof of Theorem~\ref{thm:main}}

Apply the block transfer inequality to $(W,X^n,Y^n)$:
\begin{equation}
 I(W;Y^n)\leq\min\{H(Y^n),H(Y^n\mid X^n)+\epsilon_n\}. \label{eq:app-main-block}
\end{equation}
The product law gives $H(Y^n)=nH(Y)$ and
$H(Y^n\mid X^n)=nH(Y\mid X)$, proving~\eqref{eq:main-info}. Substitution in Lemma~\ref{lem:mi-gen} gives
\begin{align}
 |\E[\gen]|
 &\leq\sqrt{\frac{2\sigma^2}{n}
        \min\{nH(Y),nH(Y\mid X)+\epsilon_n\}}\nonumber\\
 &=\sqrt{2\sigma^2\min\left\{H(Y),H(Y\mid X)+\frac{\epsilon_n}{n}\right\}}, \label{eq:app-main-final}
\end{align}
which is~\eqref{eq:main-gen}. For a loss in $[a,b]$, Hoeffding's lemma gives $\sigma^2=(b-a)^2/4$.

\subsection*{D. Proof of Proposition~\ref{prop:obstruction}}

Because $N^n$ is independent of $X^n$,
\begin{equation}
 I(W;X^n)=I(N^n;X^n)=0.                              \label{eq:app-obstruction-private}
\end{equation}
Since $N^n$ is a deterministic coordinate of $Y^n=(X^n,N^n)$,
\begin{equation}
 I(W;Y^n)=I(N^n;X^n,N^n)=H(N^n)=nH(N).               \label{eq:app-obstruction-task}
\end{equation}
Independence also gives $H(Y\mid X)=H(X,N\mid X)=H(N)$. Thus the residual branch is attained by this release at zero leakage. The calculation concerns task information only; no loss has been selected, so it is not a lower bound on every learner's generalization error.

\subsection*{E. Proof of Theorem~\ref{thm:compression}}

Theorem~\ref{thm:main} supplies
\begin{equation}
 I(W;Y^n)\leq\min\{nH(Y),nH(Y\mid X)+\epsilon_n\}.    \label{eq:app-comp-privacy}
\end{equation}
The additional hypothesis independently supplies $I(W;Y^n)\leq\rho_n$. Intersect the caps and apply Lemma~\ref{lem:mi-gen} to obtain~\eqref{eq:compression}. If the available condition is
$I(W;X^n,Y^n)\leq\rho_n$, data processing under the projection $(X^n,Y^n)\mapsto Y^n$ gives the required task budget.

If $\rho_n=o(n)$, the normalized task branch tends to zero. If $\rho_n=O(\log n)$, the displayed certificate is $O(\sqrt{\log n/n})$ independently of $H(Y\mid X)$.

\subsection*{F. Proof of Theorem~\ref{thm:individual} and the Block Reduction}

For each coordinate, Lemma~\ref{lem:transfer} gives
\begin{align}
 I(W;Y_i)
 &\leq\min\{H(Y),H(Y\mid X)+I(W;X_i)\}\nonumber\\
 &\leq\hres_{\epsilon_i}(P_{XY}).                     \label{eq:app-coordinate-transfer}
\end{align}
Insert~\eqref{eq:app-coordinate-transfer} into~\eqref{eq:app-ismi} to prove~\eqref{eq:individual}.

For the homogeneous block reduction, independence of $X_1,\ldots,X_n$ yields
\begin{align}
 I(W;X^n)
 &=H(X^n)-H(X^n\mid W)\nonumber\\
 &=\sum_{i=1}^nH(X_i)-H(X^n\mid W)\nonumber\\
 &\geq\sum_{i=1}^n[H(X_i)-H(X_i\mid W)]
 =\sum_{i=1}^nI(W;X_i).                              \label{eq:app-block-sum}
\end{align}
The function $u\mapsto\hres_u(P_{XY})$ is nonnegative, concave, and nondecreasing. The square root is also concave and nondecreasing, so
$g(u)=\sqrt{\hres_u(P_{XY})}$ is concave and nondecreasing. Jensen's inequality and~\eqref{eq:app-block-sum} imply
\begin{align}
 \frac1n\sum_{i=1}^ng(I(W;X_i))
 &\leq g\left(\frac1n\sum_{i=1}^nI(W;X_i)\right)\nonumber\\
 &\leq g(\epsilon_n/n).                              \label{eq:app-jensen}
\end{align}
Multiplication by $\sqrt{2\sigma^2}$ reproduces~\eqref{eq:main-gen}.

\subsection*{G. Proof of Corollary~\ref{cor:tv}}

Let $P$ and $Q$ be distributions on the support of $Q$, with
$q_{\min}=\min_xQ(x)>0$. The inequality $\log u\leq u-1$ gives
\begin{align}
 D(P\|Q)
 &\leq\sum_x\frac{(P(x)-Q(x))^2}{Q(x)}\nonumber\\
 &\leq\frac{1}{q_{\min}}\sum_x|P(x)-Q(x)|^2
 \leq\frac{\|P-Q\|_1^2}{q_{\min}}.                 \label{eq:app-tv-kl}
\end{align}
Use $P=P_{X_i\mid W=w}$, $Q=P_X$, and $q_{\min}=p_{\min}$. Condition~\eqref{eq:tv-condition} gives
\begin{equation}
 D(P_{X_i\mid W=w}\|P_X)\leq\delta_i^2/p_{\min}     \label{eq:app-tv-pointwise}
\end{equation}
for every supported $w$. Averaging over $W$ yields
$I(W;X_i)\leq\delta_i^2/p_{\min}$. Theorem~\ref{thm:individual} with these coordinate budgets proves~\eqref{eq:tv-bound}.

\subsection*{H. Proof of Proposition~\ref{prop:rates}}

Set
\begin{equation}
 b_n:=\min\{H(Y),H(Y\mid X)+a_n\}.                   \label{eq:app-bn}
\end{equation}
The minimum of finitely many continuous functions is continuous. Thus $a_n\to a$ implies
$b_n\to\min\{H(Y),H(Y\mid X)+a\}$, and continuity of the square root proves~\eqref{eq:limit}.

Assume $H(Y)>0$. Since $a_n\geq0$, the privacy-only sequence $b_n$ tends to zero only if $H(Y\mid X)=0$ and $a_n\to0$; those conditions are also sufficient. With task compression, the information cap includes $\rho_n/n$ as a nonnegative branch, so $\rho_n/n\to0$ is sufficient for a vanishing certificate regardless of the other terms.

\subsection*{I. Proof of Corollary~\ref{cor:adaptive} and Composition Details}

The mutual-information chain rule gives
\begin{align}
 I(\Pi_K;X^n)
 &=\sum_{k=1}^KI(V_k;X^n\mid\Pi_{k-1})
 \leq\sum_{k=1}^Ke_{n,k}=E_{n,K}.                    \label{eq:app-adaptive-x}
\end{align}
If the task ledger is present, the identical calculation yields
$I(\Pi_K;Y^n)\leq R_{n,K}$. Because $W_K$ is generated from the transcript without further data access,
\begin{equation}
 I(W_K;X^n)\leq I(\Pi_K;X^n),\qquad
 I(W_K;Y^n)\leq I(\Pi_K;Y^n).                        \label{eq:app-adaptive-dpi}
\end{equation}
Apply Theorems~\ref{thm:main} and~\ref{thm:compression} with the cumulative budgets. This proves~\eqref{eq:adaptive-one} and its direct-task refinement.

For the synergy diagnostic, $V_1=R$ is independent of $B$ and $V_2=B\oplus R$ is uniform independently of $B$, so the two marginal informations vanish. Given $V_1=R$, however, $V_2$ is a bijective function of $B$ and
\begin{equation}
 I(V_2;B\mid V_1)=H(B\mid V_1)-H(B\mid V_1,V_2)=H(B). \label{eq:app-synergy}
\end{equation}
This also equals $I(V_1,V_2;B)$ and proves that marginal per-update quantities cannot replace the conditional ledger.

The schedule formulas follow from elementary sums. Constant costs give $e_0K$. For $0<\beta<1$,
\begin{equation}
 \sum_{k=1}^Ke_0\beta^{k-1}
 =e_0\frac{1-\beta^K}{1-\beta}\leq\frac{e_0}{1-\beta}.              \label{eq:app-geometric}
\end{equation}
The integral comparison gives $\sum_{k=1}^K1/k\leq1+\log K$, while the $p$-series converges for $p>1$.

Under~\eqref{eq:full-ledger}, conditional data processing gives
\begin{align}
 I(V_k;X^n\mid\Pi_{k-1})&\leq c_{n,k},\nonumber\\
 I(V_k;Y^n\mid\Pi_{k-1})&\leq c_{n,k}.               \label{eq:app-full-ledger}
\end{align}
The cumulative task branch $\sum_kc_{n,k}/n$ is no larger than the residual branch containing the same sum plus $H(Y\mid X)$. This proves~\eqref{eq:full-ledger-bound}.

\subsection*{J. Binary Calculations}

Uniformity of $X$ and independence of the Bernoulli noise imply that $Y=X\oplus N$ is uniform. Hence $H(Y)=\log2$. Given $X$, the map $N\mapsto Y$ is bijective, so
$H(Y\mid X)=H(N)=\hb(\theta)$. Hoeffding's lemma gives $\sigma^2=1/4$ for a $[0,1]$-valued loss. Substitution in~\eqref{eq:main-gen} proves~\eqref{eq:binary-bound}.

The quoted numerical values follow from direct evaluation of
$\hb(\theta)=-\theta\log\theta-(1-\theta)\log(1-\theta)$, with $0\log0=0$. Since
\begin{equation}
 \hb'(\theta)=\log\frac{1-\theta}{\theta}>0,qquad 0<\theta<1/2,      \label{eq:app-binary-derivative}
\end{equation}
the privacy-only residual certificate increases monotonically from zero to the entropy cap as task noise moves from zero to one half.

\subsection*{K. Proofs of Theorem~\ref{thm:transfer-tight} and Proposition~\ref{prop:gen-tight}}

\emph{Deterministic-task construction:}
Assume $Y=f(X)$ and $0<H(Y)<\infty$. Let $B\sim\operatorname{Bernoulli}(\alpha)$ be independent of $(X,Y)$, introduce an erasure symbol $\mathtt e\notin\mathcal Y$, and define
\begin{equation}
 U=\begin{cases}
 Y,&B=1,\\
 \mathtt e,&B=0.
 \end{cases}                                           \label{eq:app-erasure}
\end{equation}
Conditional on $U=\mathtt e$, the law of $Y$ remains $P_Y$, whereas a nonerased output determines $Y$. Therefore
\begin{align}
 H(Y\mid U)&=(1-\alpha)H(Y),\nonumber\\
 I(U;Y)&=\alpha H(Y).                                 \label{eq:app-erasure-iy}
\end{align}
The Markov chain $X-Y-U$ and the fact that $Y$ is deterministic given $X$ imply
\begin{equation}
 I(U;X)=I(U;Y)=\alpha H(Y).                           \label{eq:app-erasure-ix}
\end{equation}
For $0\leq\epsilon\leq H(Y)$, select
$\alpha=\epsilon/H(Y)$. Then
\begin{align}
 I(U;X)&=\epsilon,\nonumber\\
 I(U;Y)&=\epsilon
 =\min\{H(Y),H(Y\mid X)+\epsilon\}.                 \label{eq:app-erasure-tight}
\end{align}
For $\epsilon\geq H(Y)$, choose $U=Y$. This proves~\eqref{eq:det-tight}.

At blocklength $n$, the same idea may be applied to the entire task block. Erasing $Y^n$ with probability $1-\alpha_n$ gives
$I(U;X^n)=I(U;Y^n)=\alpha_n nH(Y)$. Choosing
$\alpha_n=\min\{1,\epsilon_n/[nH(Y)]\}$ attains the block cap in~\eqref{eq:main-info} whenever $Y^n$ is determined by $X^n$.

\emph{Independent-residual construction:}
Let $Y=(X,N)$ with $N\perp X$. If $H(X)=0$, taking $U=Y$ attains the entropy cap with zero leakage. Suppose henceforth that $H(X)>0$. For $0\leq\epsilon\leq H(X)$, set
$\alpha=\epsilon/H(X)$ and let $V$ equal $X$ with probability $\alpha$ and an erasure symbol otherwise, using a coin independent of $(X,N)$. Take $U=(N,V)$. Since $N$ is independent of $(X,V)$,
\begin{align}
 I(U;X)&=I(N,V;X)=I(V;X)=\alpha H(X)=\epsilon,\label{eq:app-residual-ix}\\
 I(U;Y)&=I(N,V;X,N)=H(N)+\alpha H(X)\nonumber\\
 &=H(Y\mid X)+\epsilon.                             \label{eq:app-residual-iy}
\end{align}
Here $H(Y)=H(X)+H(N)$, so this is precisely the unsaturated branch of~\eqref{eq:residual-tight}. If $\epsilon\geq H(X)=I(X;Y)$, choose $U=Y$; then $I(U;X)=H(X)\leq\epsilon$ and $I(U;Y)=H(Y)$. At blocklength $n$, always reveal $N^n$ and erase or reveal the entire $X^n$ block with probability
$\alpha_n=\min\{1,\epsilon_n/[nH(X)]\}$. The same calculation gives
\begin{equation}
 I(U;Y^n)=\min\{nH(Y),nH(Y\mid X)+\epsilon_n\},
\end{equation}
which completes the proof of Theorem~\ref{thm:transfer-tight}.

\emph{Noisy-majority construction:}
For odd $n$, $S_n$ is nonzero and symmetry makes $T_n=\operatorname{sign}(S_n)$ uniform. The Markov chain $Y^n-T_n-W$ and the fact that $T_n$ is a deterministic function of $Y^n$ give
\begin{align}
 I(W;Y^n)&=I(W;T_n)\nonumber\\
 &=\log2-\hb\!\left(\frac{1-\eta}{2}\right)=d(\eta). \label{eq:app-majority-info}
\end{align}
Because $X^n=Y^n$, this also equals $I(W;X^n)$. Moreover,
$L_P(w)=1/2$ and
\begin{equation}
 L_n(w)=\frac12-\frac{wS_n}{2n}.
\end{equation}
Since $\E[W\mid T_n]=\eta T_n$,
\begin{equation}
 \E[\gen]=\frac{\E[WS_n]}{2n}
 =\frac{\eta}{2n}\E|S_n|,
\end{equation}
proving~\eqref{eq:noisy-majority-gap}. Finally,
\begin{equation}
 d(\eta)=\frac{\eta^2}{2}+O(\eta^4),\qquad
 \E|S_n|=\sqrt{\frac{2n}{\pi}}\,(1+o(1)).            \label{eq:app-majority-asymptotics}
\end{equation}
Substitution yields~\eqref{eq:rate-tight}. For this bounded loss
$\sigma^2=1/4$ and $H(Y\mid X)=0$, so~\eqref{eq:main-gen} equals
$\sqrt{\epsilon_n/(2n)}$. Their ratio tends to $\sqrt{2/\pi}$, proving Proposition~\ref{prop:gen-tight}.
For the coordinatewise claim, exchangeability gives
\begin{equation}
 \E[WY_i]=\eta\E[T_nY_i]
 =\frac{\eta}{n}\E[T_nS_n]
 =\frac{\eta}{n}\E|S_n|=q_n.
\end{equation}
Both $W$ and $Y_i$ are uniform binary variables, so their mutual information is $d(q_n)$. The same holds with $X_i$ because $X_i=Y_i$. Since all coordinatewise privacy budgets equal $d(q_n)$ and $H(Y\mid X)=0$, the right-hand side of~\eqref{eq:individual} reduces to $\sqrt{d(q_n)/2}$. As $q_n=O(n^{-1/2})$ and $d(q_n)=q_n^2/2+O(q_n^4)$, this is $(1+o(1))q_n/2$, exactly the asymptotic value of~\eqref{eq:noisy-majority-gap}.

\subsection*{L. Nonuniform Coordinate Budgets}

The individual bound reveals how an uneven privacy profile changes the certificate. In the deterministic regime $H(Y\mid X)=0$, and before the entropy cap is reached at any coordinate,~\eqref{eq:individual} reduces to
\begin{equation}
 B_{\rm ind}(\epsilon_1,\ldots,\epsilon_n)
 =\frac{\sqrt{2\sigma^2}}{n}\sum_{i=1}^n\sqrt{\epsilon_i}.           \label{eq:app-nonuniform}
\end{equation}
Fix a total certified leakage $\sum_i\epsilon_i=E$. Cauchy--Schwarz yields
\begin{equation}
 B_{\rm ind}\leq\sqrt{\frac{2\sigma^2E}{n}},          \label{eq:app-nonuniform-upper}
\end{equation}
with equality when the coordinate budgets are equal. If the leakage is evenly distributed across only $m$ active records, with $\epsilon_i=E/m$ on those records and zero elsewhere, then
\begin{equation}
 B_{\rm ind}=\frac{\sqrt{2\sigma^2mE}}{n},            \label{eq:app-m-active}
\end{equation}
provided $E/m\leq H(Y)$. Thus the block reduction is the worst certificate among allocations with the same total in this unsaturated deterministic case; a concentrated influence profile can be certified more sharply by individual information.

For a general residual, the same ordering follows from concavity of
$g(u)=\sqrt{\hres_u(P_{XY})}$. Majorization therefore formalizes the comparison: making a fixed budget vector more uniform cannot decrease the sum in~\eqref{eq:individual}. Saturation at $I(X;Y)$ creates flat coordinates, so strictness may disappear after a budget crosses the entropy-cap threshold.

The posterior-variation result has a parallel interpretation. If $H(Y\mid X)=0$ and all $\delta_i=\delta_n$ remain below saturation,~\eqref{eq:tv-bound} becomes
\begin{equation}
 |\E[\gen]|
 \leq \delta_n\sqrt{\frac{2\sigma^2}{p_{\min}}}.     \label{eq:app-tv-rate}
\end{equation}
Consequently a shrinking posterior displacement $\delta_n=O(n^{-\gamma})$ gives an $O(n^{-\gamma})$ certificate. The dependence on $p_{\min}$ reflects the conversion from an absolute posterior perturbation to relative entropy near rare prior symbols.

\subsection*{M. Polynomial Privacy and Update Regimes}

The rate classification can be made explicit for polynomial budgets. Suppose $H(Y\mid X)=0$ and
\begin{equation}
 \epsilon_n=c n^\alpha,qquad c>0.                    \label{eq:app-poly-budget}
\end{equation}
Before the entropy cap becomes active,
\begin{equation}
 |\E[\gen]|
 \leq\sqrt{2\sigma^2c}\,n^{(\alpha-1)/2}.            \label{eq:app-poly-rate}
\end{equation}
Thus $\alpha=0$ gives the parametric $n^{-1/2}$ rate, $0<\alpha<1$ gives slower convergence, and $\alpha=1$ produces a constant certificate. If $\alpha>1$, the entropy cap is eventually active and the privacy route carries no asymptotic rate information. When $H(Y\mid X)>0$, the same polynomial term approaches zero for $\alpha<1$, but the certificate approaches the residual floor.

Direct task compression obeys the identical classification with $\rho_n$ replacing $\epsilon_n$ and without a residual. If
$\rho_n=d n^\beta$ for $0\leq\beta<1$, the task route has order
$n^{(\beta-1)/2}$. The better asymptotic exponent between a deterministic privacy route and task compression is selected by the minimum.

For adaptive learning with constant conditional cost $e_0$ and a horizon $K_n=n^\kappa$, the deterministic-task version of~\eqref{eq:adaptive-one} gives
\begin{equation}
 |\E[\gen(W_{K_n},Y^n)]|
 \leq\sqrt{2\sigma^2e_0}\,n^{(\kappa-1)/2}.           \label{eq:app-update-rate}
\end{equation}
The certificate vanishes for $\kappa<1$, remains constant at $\kappa=1$, and reaches the entropy cap for sufficiently rapid growth. These statements assume that every visible component of each update is included in the conditional cost.

\subsection*{N. Conditional Data Processing Used in the Ledger}

The post-processing step in~\eqref{eq:app-adaptive-dpi} can be verified directly. Let $R_0$ be the additional random seed used to produce $W_K$ from $\Pi_K$, with
$R_0\perp(X^n,Y^n)\mid\Pi_K$. Then
\begin{align}
 I(W_K;X^n)
 &\leq I(W_K,\Pi_K;X^n)\nonumber\\
 &=I(\Pi_K;X^n)+I(W_K;X^n\mid\Pi_K)\nonumber\\
 &=I(\Pi_K;X^n).                                     \label{eq:app-postprocess-x}
\end{align}
The last term vanishes because $W_K$ is generated from $(\Pi_K,R_0)$ without new data access. The same calculation applies to $Y^n$. It also explains the scope of the corollary: if the final post-processing queries the sample again, its conditional information cost must be represented by another transcript component before data processing can be invoked.

Conditional mutual information in the ledger is evaluated under the actual joint law induced by earlier adaptive updates. The deterministic numbers $e_{n,k}$ and $r_{n,k}$ are upper bounds on those law-dependent quantities. The update kernels may depend arbitrarily on the realized transcript; no independence between releases is needed.

\subsection*{O. Verification of the Full-Sample Identity}

The warning in~\eqref{eq:full-sample-identity} follows from the chain rule:
\begin{equation}
 I(W;X^n,Y^n)=I(W;X^n)+I(W;Y^n\mid X^n).              \label{eq:app-full-identity}
\end{equation}
The second term is nonnegative and can equal $H(Y^n\mid X^n)$. In the obstruction construction $W=N^n$ and $Y^n=(X^n,N^n)$, it is exactly $nH(N)$ while the first term is zero. Hence a bound on $I(W;X^n)$ cannot be used as a bound on full-sample information unless an additional condition controls the conditional task channel.

Conversely, a full-sample budget $I(W;X^n,Y^n)\leq c_n$ implies both
$I(W;X^n)\leq c_n$ and $I(W;Y^n)\leq c_n$ by data processing. This asymmetry is why full-sample compression is sufficient for selective privacy and task compression at the information level, while selective privacy alone is not sufficient for full-sample stability.

\subsection*{P. Proof of Corollary~\ref{cor:inverse-design}}

Write $\tau_\gamma=\gamma^2/(2\sigma^2)$. Since
$\tau_\gamma<H(Y)$, the entropy branch cannot by itself make the right-hand side of~\eqref{eq:main-gen} at most $\gamma$. Squaring that bound and dividing by $2\sigma^2$ shows that the privacy branch certifies the target exactly when
\begin{equation}
 H(Y\mid X)+\frac{\epsilon_n}{n}\leq\tau_\gamma.    \label{eq:app-inverse-one}
\end{equation}
Because $\epsilon_n$ is nonnegative,~\eqref{eq:app-inverse-one} is feasible only if $H(Y\mid X)\leq\tau_\gamma$; rearrangement then gives~\eqref{eq:inverse-privacy}.

Under~\eqref{eq:task-rate}, Theorem~\ref{thm:compression} replaces the information term by the minimum of the privacy branch and $\rho_n/n$. A minimum of two real numbers is at most $\tau_\gamma$ exactly when at least one of them is at most $\tau_\gamma$. The second alternative is precisely~\eqref{eq:inverse-task}.

Finally, set $\epsilon_n=\epsilon$. When
$\tau_\gamma>H(Y\mid X)$, solving~\eqref{eq:app-inverse-one} for the positive integer $n$ gives~\eqref{eq:inverse-sample}. If
$\tau_\gamma<H(Y\mid X)$, even zero leakage leaves the privacy branch above the target. At equality, a strictly positive fixed leakage also prevents certification for every finite $n$. This proves the stated impossibility and completes the inversion.

%\clearpage

\end{document}